\documentclass[final,5p,times,twocolumn,authoryear]{elsarticle}

\usepackage{amsmath}
\usepackage[nopatch]{microtype}
\usepackage{booktabs}
\usepackage{graphics}
\usepackage{listings}
\usepackage{pgfplots}
\usepackage{algorithm2e}
\usepackage{subcaption}
\usepackage{multicol}
\usepackage{blindtext}
\usepackage{tablefootnote}
\usepackage{amsthm}
\usepackage{amssymb}
\usepackage{algorithm2e}
\usepackage{tabularx}
\usepackage[hidelinks]{hyperref}

\newtheorem{theorem}{Theorem}
\newcommand{\pluseq}{\mathrel{{+}{=}}}

\begin{document}

\begin{frontmatter}
	\title{A signal dedispersion algorithm for imaging-based transient searches}
	%\title{Streaming high time-resolution imaging dedispersion (STRIDE)}
	
	\address[pawseyaff]{Pawsey Supercomputing Research Centre, Kensington, 6151, Western Australia, Australia}
	\address[curtinaff]{International Centre for Radio Astronomy Research, Curtin University, Bentley, 6102, Western Australia, Australia}
	\address[aussrcaff]{Australian SKA Regional Centre (AusSRC), The University of Western Australia, 35 Stirling Highway, Crawley WA 6009, Australia}
	\address[skaoaff]{Square Kilometre Array Observatory (SKAO), Kensington, 6151, Western Australia, Australia}
	
	\author[pawseyaff,curtinaff]{Cristian Di Pietrantonio}
	\ead{cristian.dipietrantonio [at] csiro.au}
	
	\author[pawseyaff,aussrcaff]{Marcin Sokolowski}
	% \alsoaffiliation{Joint first authors}
	
	%% From here, in alphabetical order
	
	\author[pawseyaff]{Christopher Harris}
	
	\author[curtinaff,skaoaff]{Danny C. Price}
	
	\author[curtinaff,skaoaff]{Randall Wayth}

	\begin{abstract}
		Dedispersion is the computational process of correcting for the frequency-dependent time delay affecting a radio signal that propagates through the interstellar and intergalactic media. It is a crucial component of transient search pipelines that maximises the signal-to-noise ratio, especially when targeting highly dispersed signals: for instance, pulsar emissions making their way through a dense cloud of ionised gas, and fast radio bursts travelling cosmological distances. This paper introduces Streaming high Time-Resolution Imaging DEdispersion (STRIDE), a novel dedispersion algorithm to generate per-pixel dedispersed time series from high time and frequency resolution interferometric images. Unlike straightforward approaches to image dedispersion, STRIDE does not involve expensive manipulation of the input data layout, such as explicitly building dynamic spectra or shifting images. Furthermore, it is the first dedispersion algorithm to partition a dispersive sweep over the time dimension, in addition to frequency. As a consequence, images corresponding to the entire time span of the target dispersive delay are not required all at once. Instead, the algorithm works with an arbitrarily-sized subset of images at a time, adopting an incremental, streaming-based approach to dedispersion. In evaluating STRIDE on the presented test case, it is shown that the minimum memory requirement is reduced by 97.9\%, going from 684.5\,GB to 14.4\,GB. As current and future generations of widefield interferometers increasingly turn to imaging techniques for detection and localisation of radio transients, STRIDE positions itself as a strong alternative to traditional dedispersion methodologies. It arguably is the only viable option for imaging-based searches with low-frequency instruments such as the Murchison Widefield Array (MWA) and low-frequency Square Kilometre Array (SKA-Low).
		
	\end{abstract}

\end{frontmatter}
% You are required to provide a concise and factual abstract which does not exceed 250 words. The abstract should briefly state the purpose of your research, principal results and major conclusions. Some guidelines:

% TODO: fine a name for the algorithm!

%STRIDE

%Suggests forward motion and incremental progress.

\section{Introduction}\label{sec:introduction}

% Introduce the problem: signal dispersion through frequencies
Interferometric imaging techniques present a computational advantage of at least two orders of magnitude compared to beamforming (\cite{DiPietrantonio2025};  \cite{GPUImager}; Section 3.4 in \cite{Price2024}). Hence, these are being investigated for or already adopted in non-targeted searches for fast radio bursts (FRBs) and other transients, in conjunction with widefield interferometers and all-sky transient monitors \citep{Sherman2025,Wang2025,Sokolowski2022, Sokolowski2021, Thyagarajan2017}. A sequence of high time resolution snapshot images can be used to generate time series associated with each pixel, representing the sky brightness over time in the corresponding line of sight. These are then searched for transient events occurring anywhere in the imaged field of view (FoV).

\begin{figure}[h]
	\includegraphics[width=0.9\linewidth,trim={0, 2.5em, 0, 4em}]{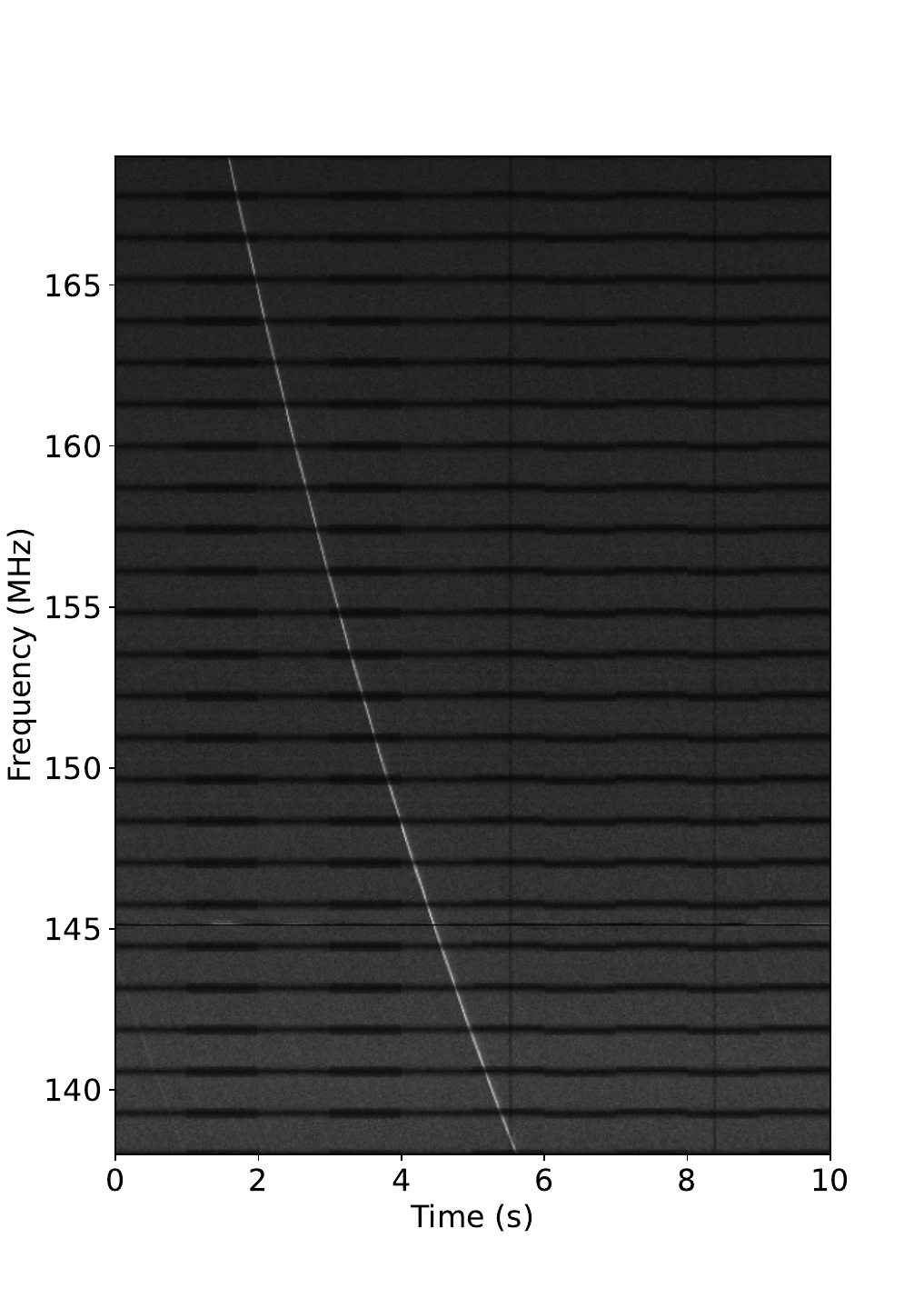}
	\caption{\textbf{A giant pulse from the Crab pulsar captured with the MWA}. A dynamic spectrum plots the observed signal intensity as a function of frequency and arrival time. Astronomical signals manifest themselves as quadratic sweeps because of the dispersive delay caused by the interstellar and intergalactic media. The Crab pulsar has a DM of approximately 57\, pc\,cm$^{-3}$. The dispersive delay across the 30.72\, MHz band centred at 154.25\, MHz is 4 seconds. The giant pulse and the fainter ones pictured in the dynamic spectrum were detected during a non-targeted transient search test using the algorithm presented in this work and implemented in the BLINK imaging pipeline \citep{DiPietrantonio2025, GPUImager}.}
	\label{fig:crab_ds}
\end{figure}

Within this framework, dedispersion techniques are essential to capture radio signals such as FRBs, whose frequency components are dispersed in time while moving along the propagation path. Because the dispersion effect grows quadratically as the signal moves towards lower frequencies, the time it takes a signal to traverse a wide band at frequencies below 300\,MHz, as observed by the Murchison Widefield Array (MWA), is in the order of tens of seconds. Figure \ref{fig:crab_ds} demonstrates the effect on a giant pulse of B0531+21, also known as the Crab pulsar. Furthermore, transients of extragalactic origin are characterised by a large dispersion measure (DM), ranging from 100's to 1000's of pc cm$^{-3}$ due to cold plasma within the intergalactic medium, exacerbating the dispersive delay. These factors make dedispersion computationally challenging.

% Chris comment: exacerbating is not a common word

Furthermore, in a non-targeted search for radio transients, also referred to as a blind search, one does not know their location, when they occur, nor their associated dispersion measure. The dedispersion process must then be performed for every DM value in a set of DM trials, for every pixel, and for every time bin, while at the same time considering tens of seconds worth of observation at any given time to account for the dispersive delay.

% cold plasma dispersion delays lower tadio frequencies by t \prop DMv^{-2} 
%	requires compensating this chromatic delay before detection

% mention lower frequencies? Mention images as input are unconventional for dedispersion algorithms?
This paper presents a novel approach to signal dedispersion in the image domain: Streaming high Time-Resolution Imaging DEdispersion (STRIDE). The algorithm iterates over batches of images referred to as image sets, capturing the FoV at different frequencies and time bins, and incrementally computes dedispersed time series for each pixel and DM trial. Based on the brute-force dedispersion algorithm, this method aims at minimising memory allocation and data movement requirements, making dedispersion efficient and effective in the image domain.

The paper is structured as follows. Section \ref{sec:related_work}  provides an overview of the state of the art on dedispersion. Section \ref{sec:dedisp_problem} formalises the incoherent dedispersion problem, and Section \ref{sec:motivation} motivates the need for a new approach. Section \ref{sec:math_framework} lays out the mathematical formulation necessary to define the novel algorithm. Then, the paper moves on to describing STRIDE in Sections \ref{sec:algorithm},  \ref{sec:ring_buffer}, and \ref{sec:generalisation}. Section \ref{sec:parallelisation} highlights the compute parallelisation strategy adopted in the authors' implementation of the algorithm. STRIDE has been evaluated on a test case and experiment results are presented in Section \ref{sec:results}. The paper ends with a discussion and closing remarks in Section \ref{sec:conclusion}.

\section{Related work} \label{sec:related_work}

A dedispersion technique may be classified as coherent or incoherent. The first class of algorithms interprets the interstellar medium as an all-pass filter that changes the phase of a radio signal on a per-frequency basis, resulting in the characteristic frequency-dependent time delay. Applying the inverse filter on voltage data, where phase and amplitude information is retained, recovers the signal as emitted by the source  \citep{ HANKINS1975,Hankins1971}. While it has the benefit of perfectly reversing dispersion, coherent dedispersion is computationally expensive. Modern software implementations make use of GPU acceleration \citep{Zhang2024, Bassa2017, Straten2011}. % Advanced signal processing techniques based on overlap-save are also employed \citep{Zhou2024}. Further, heuristics are used to determine optimal FFT parameters  \citep{Zhang2023}, and I/O performance when operating on large volume of data is also kept in consideration \citep{Kong2023}.

Incoherent dedispersion occurs after voltages have been fine channelised and transformed to sky intensity information through beamforming-based techniques or interferometric imaging. At this stage, phase information needed for exact signal correction is lost. On the other hand, data has most likely undergone a time and frequency averaging operation that reduces the computational cost of dedispersion. A frequency-dependent time shift is applied on intensity time series of each frequency channel. While the technique corrects for inter-channel delay, the dispersive smearing within each cannot be addressed. Hence, the choice of time and frequency resolution impacts the signal-to-noise ratio of the recovered signal.

There exist two main classes of incoherent dedispersion algorithms. Firstly, brute-force approaches loop over every start time bin and DM trial, possibly in parallel, and integrate across frequency channels following the path defined by the associated dispersive delay. Multiple implementations of brute force dedispersion that leverage high-performance computing hardware are present in the literature and routinely used for transient searches \citep{Novotny2023, Sclocco2016,  Barsdell2012, Magro2011}.  Incoherent, brute-force Fourier-domain dedispersion has also been proposed \citep{Bassa2022}.

Secondly, tree-based algorithms compute dispersive sweeps across the entire band in an iterative and hierarchical manner. Starting from small groups of adjacent frequency channels, partial sweeps are computed across each of those. At each subsequent stage, these partial results are combined to produce longer sweeps across a progressively larger number of consecutive channels.

The original algorithm approximates the dispersion curve with a piece-wise linear function, resulting in a loss of sensitivity \citep{Taylor1974}. The fast discrete dispersion measure transform (FDMT) algorithm \citep{Zackay2017} follows a similar tree-based strategy but makes use of the dynamic programming technique \citep{cormen2009introduction}.  The final integrated intensity values for each DM and time bin are built bottom-up using intermediate results stored in memory.  FREDDA  \citep{fredda} is the accepted primary implementation of FDMT. ESAM is a generalisation of FDMT allowing the definition of arbitrary summation masks other than the one entailed by the dispersive delay equation \citep{Gupta2024}. Heimdall is a widely-used software implementing both brute-force and tree-based methodologies, in addition to more sophisticated hybrid solutions \citep{Barsdell2012}. 

All incoherent dedispersion methodologies work with dynamic spectra because they were designed to counter a frequency-dependent temporal effect. For blind searches of radio transients to transition from beamforming to the more efficient imaging-based techniques, a dedispersion algorithm working in the image domain is required. \citet{Zackay2017} envision two ways for FDMT to operate within an imaging pipeline. The first approach is to re-arrange intensity data to form a dynamic spectrum for every pixel. The authors concede the solution is too expensive, especially when large images are coupled with the high time and frequency resolution scales involved in fast transient searches. The second option is to execute dedispersion in the visibility space. Algorithms may exploit the fact that, for sparse interferometers, most cells in the visibility grid have a zero value. Furthermore, by performing dedispersion on visibilities, the number of 2D FFT executed is reduced if there are more frequency channels than DM trials. It comes at the cost of managing the chromatic aberration associated with a wide FoV or a wide observing band. This is the solution adopted by CRACO, the imaging-based real-time transient backend for ASKAP \citep{Wang2025}, and the realfast pipeline deployed at VLA \citep{Law2018}. As interferometer core areas become densely packed with antennas, and the number of DM trials increases, the diminishing advantages of this approach are outweighed by the complexity of the implementation.

% A synonym is feeding

This paper introduces a novel incoherent dedispersion method that ingests images and produces dedispersed time series for every pixel and DM trial. Unlike previous suggestions in the literature, it does not require reshaping the intensity array. It also does not explicitly shift images in time to correct for the delay, for every DM trial. The algorithm handles the natural output layout of imaging pipelines and minimises data movement and memory requirements.

% Incoherent dedispersion is efficient and good enough to detect bright dispersed signal.

% Imaging based techniques are computationally more efficient 

\section{Incoherent dedispersion} \label{sec:dedisp_problem}

A radio signal travelling through intergalactic and interstellar media is subject to a frequency-dependent time delay that disperses the signal power in time. The delay is a frequency-dependent function $\tau (\nu_{lo}, \nu_{hi})$,

\begin{equation}\label{eq:time_delay}
	\tau(\nu_{lo}, \nu_{hi}) = K\, DM \left( \nu_{lo}^{-2} - \nu_{hi}^{-2} \right),
\end{equation}

where $K$ is a product of physical constants, $DM$ is the integrated electron density along the line of sight, known as dispersion measure, and $\nu_{lo}$ and $\nu_{hi}$ are the lower and upper frequencies of the frequency band the delay applies to.
% Introduce quantised problem
During digital signal processing, the time and frequency dimensions are quantised into a finite number of $T$ and $F$ bins of size $\Delta t$ and $\Delta \nu$, respectively. Hence, the frequency-dependent time delay can be discretized as

\begin{equation}\label{eq:channel_delay}
	\delta (f) = \left\lceil   \frac{\tau(\nu_f,   \nu_f + \Delta \nu)}{\Delta t} \right\rceil,
\end{equation}

where $\nu_f$ is the frequency at the lower end of the $f$-th frequency channel, $f \in \mathcal{F}$, $\mathcal{F} = \{1, \ldots, F\}$. The frequency channel $F$ is at the highest end of the frequency band, whereas channel $1$ sits at the lowest one.

A dynamic spectrum $D$, $D: \mathcal{F} \times \mathcal{T} \to \mathbb{R}$, $\mathcal{T} =  \{1, \ldots, T\}$,  illustrates the signal intensity coming from a given direction in the sky as a function of frequency and time bins. The dispersion of a radio signal manifests as a quadratic sweep across the dynamic spectrum, as shown in Figure \ref{fig:crab_ds}. 

% Introduce the dynamic spectrum
To recover the time series $S(t)$ of the total signal intensity one must reverse the time delay effect while accumulating $D(f, t)$ over frequency channels:

% TODO change o to something else

\begin{equation}\label{eq:time_series_full_dedisp}
	S(t) = \sum_{f = 1}^F \sum_{d=0}^{\delta(f) - 1}  D(f, t + g(f) + d),
\end{equation}

 $t \in \mathcal{T}$, where the summation over $d$ accounts for the situation when the signal persists in frequency channel $f$ for longer than a single time bin, and  $g: \mathcal{F} \to \mathbb{N}_0$,

\begin{equation}\label{eq:cumulative_delay}
g(f) = \begin{cases}
	0, & \text{if } f = F \\
	\left\lceil  \frac{\tau(\nu_{f + 1},  \nu_F + \Delta \nu)}{\Delta t} \right\rceil - 1, & \text{otherwise,}
\end{cases}
\end{equation}

is the time, in units of time bins, the sweep takes from entering the frequency band to just before reaching channel $f$. In the present work we adopt the convention that a signal enters frequency channel $f$ at the same time bin it leaves the previous channel $f + 1$.

\begin{figure*}[h!]
	\centering
	\includegraphics[width=\linewidth]{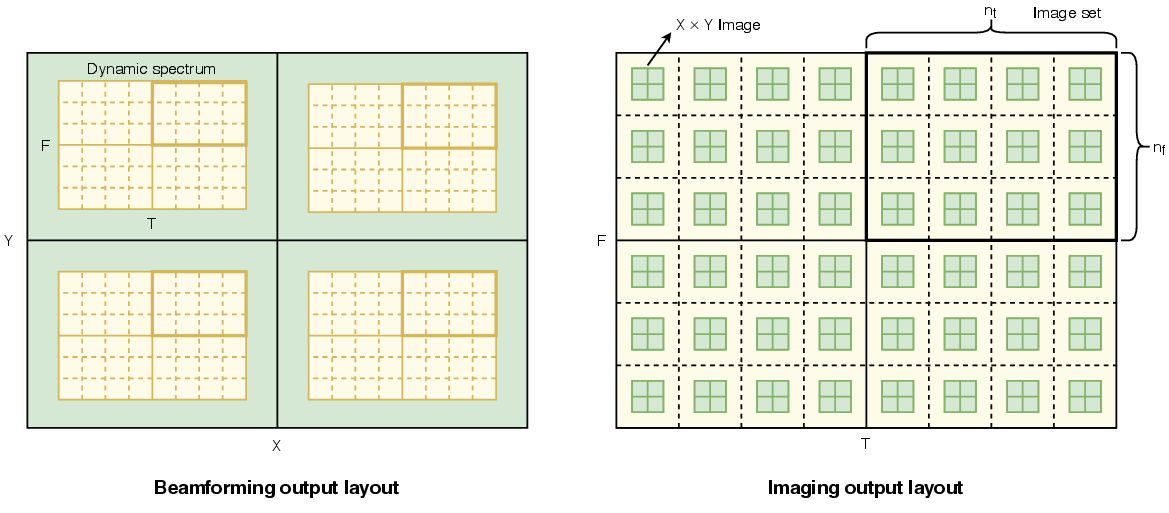}
	\caption{\textbf{Beamforming and imaging output data layouts.} Beamforming produces a dynamic spectrum of dimensions $F \times T$ for each sky pointing. Time and frequency information is contiguously stored in memory, whereas dynamic spectra do not need to be, and typically are not, held in memory at the same time (left diagram). A high time and frequency resolution imaging pipeline generates images of dimensions $X \times Y$, one for each fine channel and time bin. A pixel value within an image encodes the signal intensity at the associated frequency channel, time bin, and direction in the sky (right diagram). The number of pixels does not place a memory requirement on the processing of dynamic spectra produced through beamforming because they are processed independently. Conversely, image size poses a limitation on how many images can be simultaneously held in memory, and hence the number $n_t$ of time bins and the number $n_f$ of frequency channels that are available at any given time for dedispersion. STRIDE partitions images in image sets containing $n_t n_f$ images each and that can entirely reside in memory.}
	\label{fig:data_layout}
\end{figure*}

\section{Motivation} \label{sec:motivation}

%TODO: aving images means there is a time series for each pixel => millions of pixels. Give some example of existing interferometers for comparison. Also compare with beamforming set ups.

% introduction & assumptions
Transient searches based on beamforming \citep{Bhat2023, Bezuidenhout2022} partition the FoV in a regular grid of dimensions $X \times Y$. Each cell is covered by a tied-array beam resulting in fine-channelised time series whose time and frequency information is stored in a dynamic spectrum of dimensions $F \times T$. Hence, the beamforming output is a 4D array of dimensions $X \times Y \times F \times T$. The left diagram of Figure \ref{fig:data_layout} illustrates the described data layout. Traditional dedispersion algorithms assume a dynamic spectrum is explicitly constructed and kept in memory, at least enough of it to cover the targeted maximum dispersive delay. This is the case for beamforming searches where a dynamic spectrum has its entries placed contiguously in memory and hence it can be independently stored and processed. 

High time and frequency resolution imaging generates a snapshot image of the FoV for each time bin and frequency channel, resulting in a 4D array of dimensions $F \times T \times X \times Y$, shown in the right diagram of Figure \ref{fig:data_layout}. In other terms, sky brightness measurements across all pointings and relative to a certain time bin and frequency channel are tightly packed in memory, but time-frequency information of each pointing has a stride of $XY$ elements. The array must be reshaped to the beamforming output layout for traditional dedispersion algorithms to work. However, the operation is computationally prohibitive when dealing with hundreds of thousands of pixels per image, due to the high angular resolution of widefield interferometers, and tens of thousands of images generated by the high time and frequency resolution imaging process \citep{DiPietrantonio2025}.

% We cannot keep in memory all the images related to the dispersive delay length
% TODO add some numerical example
The problem is exacerbated at frequencies below 300\,MHz where the dispersive delay in extragalactic radio signals can be in the tens of seconds. Several hundred thousand images are needed to cover the full band and that time interval in millisecond-time and kHz-frequency resolution. The resulting volume of data exceeds the allocatable memory available on most computing systems. Table \ref{tab:crab_requirements} in Section \ref{sec:results} quantifies these requirements for a transient search case with the MWA. In such situations one cannot generate and hold on to all the images necessary to have a complete dynamic spectrum for every pixel.

This work presents a novel dedispersion algorithm that makes imaging-based fast transient searches at low frequencies feasible. Instead of requiring all images covering the entire frequency band and the dispersive delay to be readily available in memory, the algorithm works with a subset of those at a time, called an image set. Dedispersed time series are computed incrementally as new image sets are made available. Once an image set has been processed, it can be discarded to make space for a new one. 

An image set captures only a small 2D section of a dynamic spectrum, but it does so for all $XY$ dynamic spectra.  For instance, consider the imaging output array depicted in the right-side diagram of Figure \ref{fig:data_layout}. It is divided into four image sets, delineated with black solid lines. The top-right image set, highlighted with a bolder contour, corresponds to the top right sections of the dynamic spectra in the beamforming output, marked with bold solid lines.

There are three problems to be tackled: (a)  computing dedispersion for a given DM iteratively and incrementally, with only a portion of the dynamic spectrum made available at each iteration; (b) running the process for multiple DMs and on multiple partial dynamic spectra at the same time; and (c) efficiently dealing with time-frequency information not being contiguous in memory. The novelty and complexity of the algorithm comes from solving the first point. Then, it can be generalised to handle multiple DMs and dynamic spectra, hence pixels, addressing points (b) and (c).

\section{Mathematical framework}
\label{sec:math_framework}

\begin{figure}[h!]
	\centering
	\includegraphics[width=\linewidth]{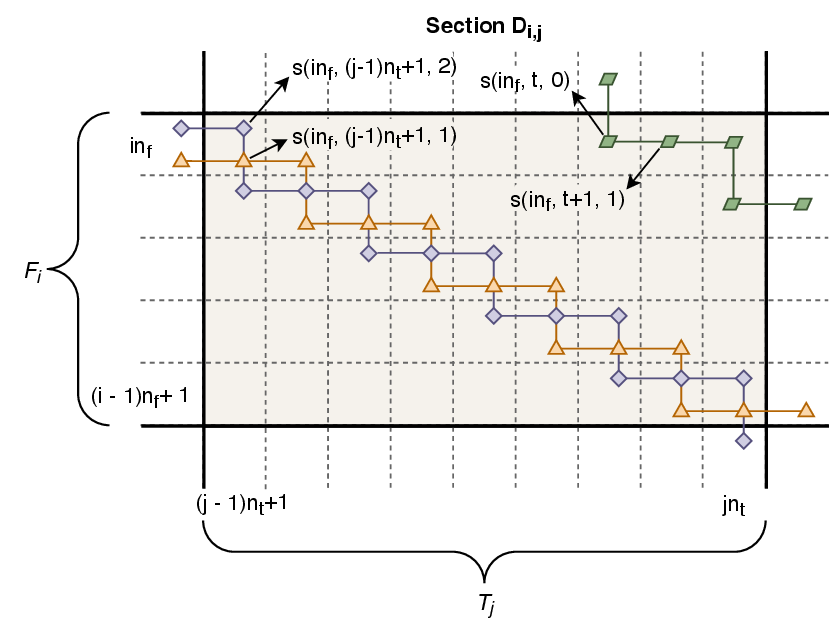}
	\caption{\textbf{Visual representation of a section.} Only a 2D portion of size $n_t \times n_f$ of a dynamic spectrum is available at any given time. It covers a subset $\mathcal{F}_i$ of the original interval $\mathcal{F}$ of frequency channels, and a range $\mathcal{T}_j$, $\mathcal{T}_j \subseteq \mathcal{T}$, of time bins. Sweeps cross the section along paths defined by the dispersive delay and accumulate intensity values associated with those. In this example, each frequency channel is characterised by a discrete time delay of 3 time bins. That is, $\delta(f) = 3,\ \forall f$. The two side sweeps of channel $in_f$ are depicted with an indigo diamond and orange triangle patterns. These have entered the channel in the time adjacent section $D_{i,j-1}$ and crossed the boundary with the current one. A top sweep, represented with a green parallelogram pattern, crosses the top right corner of the section. The sweep enters the top channel in the time interval encompassed by the section, traverses it entirely, and enters channel $in_f -1$ before leaving the section. It will become a side sweep of section $D_{i, j + 1}$.}
	\label{fig:section}
\end{figure}

This section formalises the dedispersion problem for a fixed DM in terms of partial dynamic spectra. It also introduces a new way of identifying  and classifying sweeps that cross a section of a dynamic spectrum. Then, sweeps and sections are used to define the dedispersed time series, which is the final output of the dedispersion algorithm. Table \ref{tab:list_of_symbols} in the Appendix keeps track of the notation used throughout this work.

A dynamic spectrum $D$ can be partitioned over frequency channels and time bins in two-dimensional sections $D_{i,j}$ of size $n_f \times n_t$, $i \in \{ 1, \ldots,  F/n_f \}$, $j \in \{ 1, \ldots,  T/n_t\}$. A section $D_{i,j}$ covers a subset $\mathcal{F}_i$ of the frequency channels encompassed by  $D$,  and a subset $\mathcal{T}_j$ of all time bins. Figure \ref{fig:section} illustrates the described setting.

Formally, let 
\begin{equation}\label{eq:freq_interval}
	\mathcal{F}_i = \{(i - 1)n_f + 1, \ldots, i n_f\}
\end{equation}
	and  
\begin{equation}\label{eq:time_interval}
	\mathcal{T}_j =  \{(j - 1)n_t + 1, \ldots, j n_t \}.
\end{equation}
Then,
\begin{multline}
	 D_{i,j}:  \mathcal{F}_i  \times \mathcal{T}_j \to \mathbb{R},\  D_{i,j}(f, t) = D(f, t)\\ \forall (f, t) \in  \mathcal{F}_i  \times \mathcal{T}_j.
\end{multline}

% TODO Specify which one is the low end of the frequency channel

Intensity values are accumulated on paths across the partial spectrum $D_{i,j}(f, t)$, referred to as sweeps. These are portions of complete sweeps that could be computed across $D$, if it were fully available.

%TODO create image showing the two types of sweeps within a seciton of the dynamic spectrum. Picture the dynamic spectrum grayed out except for the currently active section. Sweeps in different colours. Might be worth doing two images of the same concept: one in continuous form, and the other in discrete form. Or just one??

Any sweep entering section $D_{i,j}$ of the dynamic spectrum falls in exactly one of the following categories.

\begin{enumerate}
	\item Top sweeps that enter the top frequency channel $in_f$ of the section at a time $t$ encompassed by $D_{i,j}$. That is,  $ t \in \mathcal{T}_j$. A top sweep is illustrated in Figure \ref{fig:section} on the top right corner of the section using a green parallelogram pattern.
	\item Side sweeps that enter a channel $f$ covered by the section,  $f \in \mathcal{F}_i$, at some time $t$ previous to the start time of the section. That is, $t  < (j - 1)n_t + 1$. Two side sweeps are shown in Figure \ref{fig:section} through an indigo diamond pattern and an orange triangle one.
\end{enumerate}

There are $n_t$ top sweeps crossing a section. Furthermore, there are $\delta(f) - 1$ side sweeps for each frequency channel $f$, where $\delta(f)$ is the discrete time delay within that frequency channel in units of time bins. A proof of the statement is provided in \ref{sec:proofs}.

% Compute the corresponding start time of sweeps
A sweep $s(f, t, d)$ going across $D$ can be uniquely identified by the frequency channel $f$ it is traversing at time bin $t$, and by the number of time bins $d$ the sweep has already spent in $f$ prior to $t$. For instance, the set $A_{i,j}$ of all the top sweeps of $D_{i,j}$ is defined as

\begin{equation}\label{eq:top_sweep_set}
A_{i,j} = \{s(in_f, t, 0) : t \in \mathcal{T}_j\}.
\end{equation}
On the other hand, all side sweeps of $D_{i,j}$ are contained in the set $B_{i,j}$,

\begin{multline}\label{eq:side_sweep_set}
B_{i,j} = \{s(f, (j-1)n_t + 1, d) : \\ f \in\mathcal{F}_i\ \land\ d \in \{1, \dots, \delta(f) - 1  \} \}.
\end{multline}

Furthermore, the same sweep can be identified by different tuples. As an example, a sweep crossing channel $f$ in two time bins $t$, $t + 1$ before reaching the adjacent lower channel $f- 1$ can be identified as

\begin{equation}
	s(f, t, 0) = s(f, t + 1, 1) = s(f -1, t + 1, 0).
\end{equation}

Any sweep $s(f, t, d)$ across D and its sections enters the frequency band through the top channel $F$ at time $T(f, t, d)$,

\begin{equation}\label{eq:start_time}
	T(f, t, d) = t - d - g(f).
\end{equation}

Hence, 

\begin{equation}
	s(f, t, d) = s(F, T(f, t, d), 0).
\end{equation}

The dedispersed time series $S(t)$, $t \in \mathcal{T}$, already defined in Equation \ref{eq:time_series_full_dedisp}, can be rewritten using the new notation. For each time bin $t$, a new sweep $s(F, t, 0)$ enters $D$ through the top channel $F$ and crosses the spectrum along a predefined path. Coordinates $(p, q)$ of the entries of $D$ associated to $s(F, t, 0)$ are enumerated by exploiting the fact that the same sweep is identified by a unique tuple at each point along its path.

Formally,

\begin{equation}\label{eq:total_intensity_as_triple_sum}
	S(t) = \sum_{p, q\ :\ \exists d\ s(p, q, d) = s(F, t, 0)} D(p, q).
\end{equation}

However, only a section $D_{i,j}$ is available at a time. What can be computed is the contribution $S_{i,j}(t)$ of section $D_{i,j}$ to the total intensity $S(t)$ of a sweep $s(F, t, 0)$:

\begin{equation}	
 S_{i,j}(t) = \sum_{\substack{p, q\ :\ \exists d\ s(p, q, d) = s(F, t, 0)\ \land\\ (p, q)\ \in  \mathcal{F}_i  \times \mathcal{T}_j}} D_{i,j}(p, q).
\end{equation}

Furthermore, sweep $s(F, t, 0)$ either crosses section $D_{i,j}$ or it does not. In the former case, $s(F, t, 0)$ is either a top sweep or a side sweep of $D_{i,j}$. In the latter case, the section does not add any contribution. Finally, a sweep must cross at least one section of $D$. Then, the total intensity of a sweep entering $D$ at time bin $t$ can be expressed as the sum of contributions from all sections:

\begin{equation}\label{eq:sum_of_parts}
	S(t) = \sum_{i,j} S_{i,j}(t).
\end{equation}

%formalised as
%
%\begin{equation}
%	\exists s(x, y, d)\  s(x, y, d) = s(F, t, 0)\ \land\ s(x, y, d)\ \in A_{i,j}\cup B_{i,j}.
%\end{equation}

% now what I want to say is that the above equations account for all the contributions?

% Also one want to prove that to compute all possible I_ij(s) one as to look at side sweeps and top sweeps
% Each continuing/top sweep belongs to only one complete sweep <=> there is at most one subsweep per subregion belonging to a given complete sweep
%  

% TODO: consider introducing algo 2 (compute_partial_sweep) before algo 1

\section{The STRIDE algorithm}
\label{sec:algorithm}

\begin{figure*}
	\includegraphics[width=\linewidth]{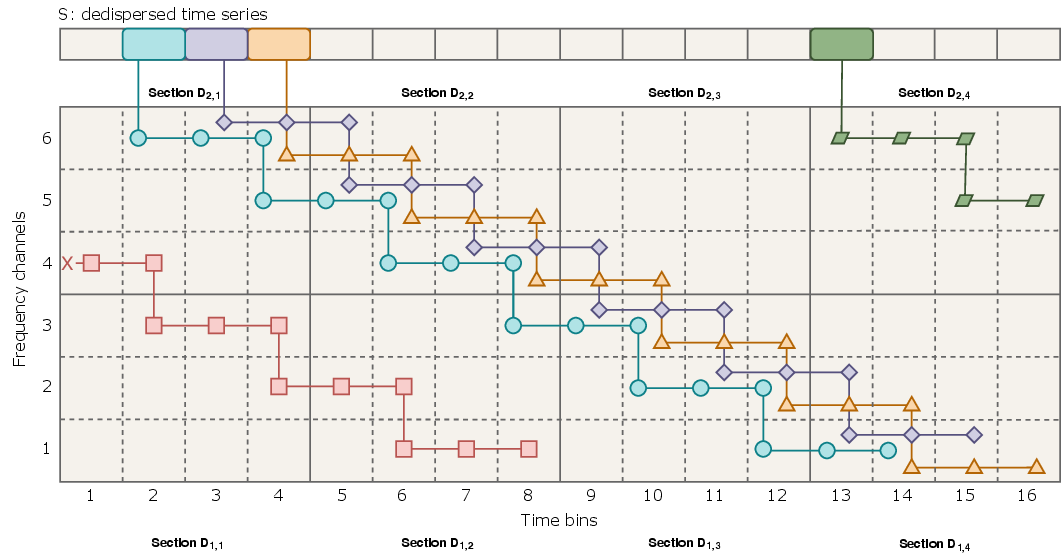}
	\caption{\textbf{Example execution of the dedispersion algorithm.} The full dynamic spectrum $D$, spanning $T = 16$ time bins and $F = 6$ frequency channels, is partitioned in 8 sections of dimensions $n_t = 4$, $n_f = 3$. Aligned on top of $D$ is the array holding the accumulated intensities across sweep paths for a fixed DM value and for all start time bins. To simplify the example, the discrete time delay is set to be the same for each frequency channel and it is equal to 3. That is, $\delta(f) = 3,\forall f\ f \in \{1, \ldots, 6\}$. Then, the cumulative delay $g(f)$ can be simplified as $g(f) = 2(6 - f)$. Algorithm \ref{algo:dedisp} is executed for each section of the dynamic spectrum. In this example, there are $n_t = 4$ top sweeps and $3(\delta(f) - 1) = 6$ side sweeps, two for each channel, in each section. Firstly, all top sweeps are computed. Three of these, $s(6, 2, 0)$, $s(6, 3, 0)$, and $s(6, 4, 0)$, are shown entering Section $D_{2,1}$ and are depicted with a blue circle, an indigo diamond, and an orange triangle pattern, respectively. The \texttt{compute\_partial\_sweep} routine follows their path through $D_{2,1}$, accumulating the associated intensity values, until they cross the section boundary. Resulting partial intensities are added to the corresponding entries in the time series array with indexes $T(6, 2, 0) = 2$,  $T(6, 3, 0) = 3$, and  $T(6, 4, 0) = 4$ . Sections whose time interval index is 1, $D_{1,1}$ and $D_{2,1}$ in this case, do not have side sweeps that can be completed. For instance, the red square sweep $s(4, 1, 1)$ traversing sections $D_{2,1}$, $D_{1,1}$, and $D_{1,2}$ entered the band at a time bin $T(4, 1, 1) = 1 - 1 - g(4) = -4$ not covered by $D$. Hence, it does not have an entry in the dedispersed time series array and its accumulated intensity is discarded. Actionable side sweeps are encountered when the algorithm processes section $D_{2,2}$. The indigo diamond sweep and the orange triangle one enter section $D_{2,2}$ as side sweeps of channel 6, and are now identified as $s(6, 5, 2)$  and $s(6, 5, 1)$. On the other hand, the circle sweep traverses channel 6 entirely in the time-adjacent section $D_{2,1}$ and enters section $D_{2,2}$ through channel 5 as $s(5, 5, 1)$. Their respective start times are $T(6, 5, 2) = 5 - 2 - g(6) = 3$,  $T(6, 5, 1) = 4$, and $T(5, 5, 1) = 2$. Intensity contributions of section $D_{2,2}$ to those sweeps are added to the corresponding array entries where contributions from previous sections are also found. The sweep length $L$ (Equation \ref{eq:sweep_length}) across the entire spectrum $D$ is defined as $L = g(1) + \delta(1) = 10 + 3 = 13$.  The circle, diamond, and triangle sweeps reach the bottom channel within the pictured time frame. The sweep $s(6, 13, 0)$, drawn with a green parallelogram pattern, does not, and requires additional future sections of the dynamic spectrum to accumulate the intensity values over its full path.}
	\label{fig:full_dyspec_repr}
\end{figure*}

Equation \ref{eq:sum_of_parts} states that the dedispersed time series can be reconstructed by iterating over a partition of the dynamic spectrum for as many times as the number of time bins. In a practical scenario one wants to iterate over the partition once, and for each section to establish its contributions to the final dedispersed time series. The latter approach allows the incremental generation and dedispersion of the dynamic spectrum in situations where processing its full extent is constrained by memory and compute capacity, or by the underlying data layout. Figure \ref{fig:full_dyspec_repr} provides a visual representation of a dynamic spectrum partitioned in sections, with example sweeps whose intensity is progressively accumulated as sections are produced. Similarly to the previous section, STRIDE is here introduced assuming a given DM and a single dynamic spectrum. The algorithm will be generalised to handle DM ranges and multiple pixels in a later section.

% Maybe give an overview of the algorithms presented in this section and the next? Also re - iterate that the generalisation to multiple DM and pixels is presented in the next section.

\begin{algorithm}
	\SetAlgoLined
	\KwIn{$D_{i,j}$, a section of the dynamic spectrum,  $i \in \{ 1, \ldots,  F/n_f \}$, $j \in \{ 1, \ldots,  T/n_t\}$.}
	\KwIn{$S$, a reference to the dedispersed time series array where partial contributions are placed in.}
	\KwResult{Contributions from $D_{i,j}$ to its top and side sweeps are added to $S$. }
	\vspace{0.5em}

   // compute all the top sweeps in section $D_{i,j}$. \\
	\For{$t$ \textbf{in} $\mathcal{T}_j$}{
		$t_s \gets T(in_f, t, 0)$\;
		\textbf{if} $t_s < 1$ \textbf{then} \textbf{continue}\;
		$S[t_s] \pluseq $ compute\_partial\_sweep($D_{i,j}$, $t$, $in_f$, $0$)\;
	}
	\vspace{0.5em}
	// now compute all the side sweeps. \\
	\For{$f$ \textbf{in} $\mathcal{F}_i$}{
		\For{$d \gets 1$ \textbf{to} $\delta(f) - 1$}{
			$t_s \gets T(f, (j-1)n_t + 1, d)$\;
			\textbf{if} $t_s < 1$ \textbf{then} \textbf{continue}\;
			$S[t_s] \pluseq $ compute\_partial\_sweep($D_{i,j}$,  $(j-1)n_t + 1$, $f$, $d$)\;
		}
	}
	\caption{The \texttt{partial\_dedispersion} routine. Given a section $D_{i, j}$ of the dynamic spectrum, the procedure iterates over all the crossing sweeps to determine contributions of the section to the final dedispersed time series, stored in $S$.}
	\label{algo:dedisp}
\end{algorithm}

Algorithm \ref{algo:dedisp}:  \texttt{partial\_dedispersion} takes as input a section $D_{i,j}$ of the dynamic spectrum and executes brute-force dedispersion to calculate contributions of the section to all the sweeps crossing it. These are all represented in $A_{i,j}\cup B_{i,j}$, the set containing all top and side sweeps. The process consists of two phases. During the first one, contributions to the intensity of top sweeps are computed. The algorithm does so by iterating over all the time bins when such sweeps enter the top channel of the section, and accumulates values along the sweep path encompassed by $D_{i,j}$. In the second phase, an outer loop over the frequency channels of the section is paired with an inner loop over all admissible values of $d$, from here on referred to as offsets, to handle side sweeps.

The time at which a top or side sweep enters the frequency band is computed before the accumulation of intensity values takes place. The sweep is ignored if said time is prior to the start time of the dedispersed series. Otherwise, the supporting function \texttt{compute\_partial\_sweep} is invoked to sum intensity values across the sweep path defined by the originating time bin, frequency channel, and offset. The result is added to the corresponding entry in the output one-dimensional array $S$, which stores the dedispersed time series.

The computation of a partial sweep is illustrated in Algorithm \ref{algo:partial_sweep}: \texttt{compute\_partial\_sweep}. Given a start time bin, frequency channel, and an initial offset, the procedure traces the sweep path by iterating over frequency channels in descending order, and relevant time bins within each of those, to accumulate the associated values in $D_{i,j}$. The algorithm terminates when the sweep crosses a boundary of the section, by moving past either the bottom frequency channel or the last time bin.

The start frequency channel $f_{top}$ is treated as a special case for the \texttt{compute\_partial\_sweep} function to handle both top and side sweeps. The number of time bins necessary for a sweep to cross $f_{top}$ within the current section depends on the associated offset $d_{start}$, the number of time bins already spent in $f_{top}$ prior to $t_{start}$. The offset is always $0$ for top sweeps, and varies between $1$ and $\delta(f) - 1$ for side sweeps.

The time bin $t$ at which the sweep enters frequency channel $f$, $f < f_{top}$, is equal to $t_{start}$ plus the number of time bins it took the sweep to cross the frequency band between $f_{top}$, included, and $f$, excluded. From the resulting value is subtracted $d_{start}$, the number of time bins spent in $f_{top}$ in the time-adjacent section $D_{i, j -1}$. The intra-channel dispersive delay $d_{max}$ is also adjusted accordingly. Finally, for each frequency channel, a for loop accumulates a number of consecutive intensity values equal to the dispersive delay within that channel.

% TODO Discussion on iteratively computing the solution, and inner parallelism?? (maybe this goes later in implementation details)

\begin{algorithm}
	\SetAlgoLined
	\KwIn{$D_{i,j}$, a section of the dynamic spectrum, $i \in \{ 1, \ldots,  F/n_f \}$, $j \in \{ 1, \ldots,  T/n_t\}$.}
	\KwIn{$t_{start}$, the start time bin of the partial sweep.}
	\KwIn{$f_{top}$, the start frequency channel of the partial sweep.}
	\KwIn{$d_{start}$, offset within $f_{top}$ the sweep starts at.}
	\KwOut{$s$, the intensity accumulated by the partial sweep.}
		\vspace{0.5em}
	$s \gets 0$\;
	\For{$f \gets f_{top}$ \textbf{down to} $(i - 1)n_f + 1$ }{
		\uIf{$f = f_{top}$}{
			$d_{max} \gets \delta(f) - d_{start}$\;
			$t \gets  t_{start}$\;
		}\Else{
			$d_{max} \gets  \delta(f)$\;
			$t \gets  t_{start} + (g(f) - g(f_{top})) - d_{start}$\;
		}
		\For{$d \gets 0$ \textbf{to} $d_{max} - 1$}{
			\textbf{if} $t + d > jn_t$ \textbf{then} \textbf{return} s\; 
			$s \pluseq D_{i,j}(f, t + d)$\;
			
		}
	}
	\Return{$s$}\;
	\caption{\texttt{compute\_partial\_sweep}. Compute a partial sweep across section $D_{i, j}$ starting from time bin $t_{start}$ and frequency channel $f_{top}$.}
	\label{algo:partial_sweep}
\end{algorithm}

\section{Ring buffer strategy} \label{sec:ring_buffer}
% considering such an array is allocated for every DM trial and for every pixel.
Algorithms \ref{algo:dedisp} and \ref{algo:partial_sweep} store the dedispersed time series in a one-dimensional output array, $S$. Its size $N$ is limited by the amount of available computer memory. What follows is an analysis to determine a suitable value for $N$.

Let $L$ be the time required for the sweep to cross the entire band, in units of time bins. Referred to as sweep length, it is defined as

\begin{equation}\label{eq:sweep_length}
	L = g(1) + \delta(1).
\end{equation}

For the time series value $S[t]$ to be the accumulated intensity of a sweep starting at time bin $t$ and crossing the entire spectrum, the algorithm must have already processed all the sections  covering the time bins $\{t, \ldots, t + L -1\}$. If $N < L$, no element of $S$ represents a complete sweep. In the case $N = L$, only $S[1]$ is a usable value. In general, $N$ must be defined as

\begin{equation}\label{eq:buffer_length}
	N = B + L,
\end{equation}

extending $S$ with an additional number $B$ of memory slots. In doing so, the first $B + 1$ elements of $S$ are total intensity values ready to be searched for peaks. The remaining $L - 1$ elements require contributions from future sections.

As is, the algorithm cannot process further time bins to complete the last entries of the time series because new sweeps starting at those times do not have corresponding slots in $S$. However, if the first $B + 1$ elements of the array were consumed in a downstream processing task, such as a peak finding routine, then those memory locations could be repurposed to collect the total intensity of new sweeps. The array $S$ becomes a ring buffer with $N$ slots. Oldest total intensity elements are consumed by the peak finding algorithm as new entries are produced by the dedispersion algorithm. 

% TODO Furthermore, $S$ is guaranteed to have enough space available for \texttt{partial\_dedispersion} to e a section  by imposing the additional requirement $B \ge n_t$, 

\begin{algorithm}
	\SetAlgoLined
	\KwIn{$\{D_{i,j}\ :\ i \in \{ 1, \ldots,  F/n_f \}\ \land\ j \in \{ 1, \ldots,  T/n_t\} \}$, a partitioned dynamic spectrum.}
	\KwIn{$B$, $B \geq n_t$, the number of additional memory slots.}
	\KwResult{Generation and processing of dedispersed time series.}
		\vspace{0.5em}
	$N \gets L + B$\;
	$S[t] \gets 0,\ \forall t\ t \in \{1,\ldots, N\}$\;
	$r_s \gets 1$; // earliest position of the buffer. \\
	$r_c \gets 0$; // number of consumed slots. \\
	\For{$j \gets 1$ \textbf{to} $T/n_t$}{
		// if the buffer is full, then process earliest elements. \\
		\If{$r_c + n_t  > N$}{
			// h is the number complete sweeps.\\
			$h \gets r_c - L + 1$\;
			peak\_finding($S$, $r_s$,  $h$)\;
			clear\_buffer($S$, $r_s$, $h$)\;
			$r_c \gets r_c - h$\;
			$r_s \gets (r_s - 1 + h) \bmod N + 1$\; 
		}
		\For{$i \gets 1$ \textbf{to} $F/n_f$}{
			partial\_dedispersion($D_{i, j}$, $S$, $r_s$, $r_c$)\;
		}
		$r_c \gets r_c + n_t$\;
	}
	// process leftover entries, if any.\\
	\If{$r_c \geq L$}{
		$h \gets r_c - L + 1$\;
		peak\_finding($S$, $r_s$,  $h$)\;
	}
	\caption{\texttt{transient\_search}. An incremental dedispersion and  transient search pipeline.}
	\label{algo:ring_buffer_algo}
\end{algorithm}

%TODO: add a note on B having to be larger than time batches , but not necessary a multiple

%TODO add an image explaining the process

Algorithm \ref{algo:ring_buffer_algo}: \texttt{transient\_search} uses partial dedispersion in conjunction with a ring buffer and a peak finding routine to detect transients within a partitioned dynamic spectrum. Figure \ref{fig:ring_buffer} provides a visualisation of the algorithm execution on a simple, concrete example. Initially, all the $N$ slots of $S$ are initialised to zero. The start position of the buffer is held in $r_s$, and the number of utilised slots is recorded in $r_c$.

The computation is driven by an outer loop over the partitioned time axis of the dynamic spectrum, such that $n_t$ time bins are processed at each iteration. The requirement on $B$, $B \geq n_t$, ensures $S$ has enough free slots to hold new sweeps. An inner loop executes the dedispersion algorithm on sections covering the current time interval and, collectively, the entire frequency band. By the end of an outer loop, the dedispersion algorithm will have populated the slots $S[r_c + 1], \ldots, S[r_c + n_t]$ with contributions from the currently considered sections to new sweeps starting at the corresponding time bins. 

Furthermore, there are sweeps that entered the frequency band at past time bins $t_s$, $t_s \leq (j - 1)n_t$, and are now intersecting sections processed at the current outer loop iteration. When $r_c > 0$, the algorithm places contributions to those into the corresponding slots, if any, whose index ranges in  $\{r_s, \ldots,  r_s + r_c - 1\}$. During the first iterations of the pipeline there may be sweeps intersecting the current sections and whose start time is less than what is associated with the earliest slot $S[r_s]$. For instance, all the side sweeps of sections $D_{i,1},\ \forall i$, generated before the first time bin of the dynamic spectrum. Those, marked in red in Figure \ref{fig:timing1}, have no associated slots and are discarded.

Dedispersion runs until $S$ is not able to hold further $n_t$ new sweeps. There are $r_c$, $r_c > L$, slots with accumulated contributions to the corresponding sweeps. The last $L - 1$ of these are still waiting for contributions from future time bins because the sweep length is $L$. The earlier $h = r_c - L + 1$ slots have accumulated all contributions and are ready for processing. The peak finding algorithm searches for peaks in those, then that section of the ring buffer is cleared by assigning zeros to its slots. The start position $r_s$ is moved $h$ slots forward and the count $r_c$ is reduced by the same amount.

From this moment onwards, the dedispersion and the peak finding routines alternate in cycle. The first one processes more time bins and fills up $S$, the latter consumes the $h$ oldest entries therein. Side sweeps that started at earlier cycles are never discarded. Because their maximum length is $L$ and at least one time bin is spent in the current time interval $\mathcal{T}_j$, contribution to side sweeps fall within the first  $L - 1$ slots in $S$, which were not cleared at the end of the past cycle.

%TODO read here https://iopscience.iop.org/article/10.3847/1538-4357/add014#apjadd014fn6

% TODO: consider replacing slots with cell in relation to the ring buffer

\begin{figure*}
	\centering
	\begin{subfigure}[t]{0.45\textwidth}
		\centering
		\includegraphics[width=\linewidth]{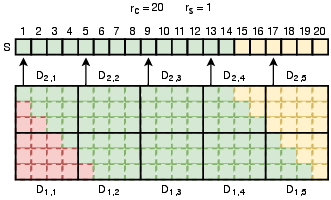} 
		\caption{\textbf{End of cycle 1}. Sections related to the first 5 time bin intervals, encompassing time bins 1 to 20, are fed to the \texttt{partial\_dedispersion} algorithm. The figure highlights contributions from selected top sweeps by means of arrows. These generate from the top edge of the sections, at time bins when the edge is crossed by the associated sweeps, and point to the respective buffer slots. All side sweeps from sections $D_{1,1}$ and $D_{2,1}$ are discarded as their start times predate the first time bin. Elements of the dynamic spectrum that do not belong to any valid sweep are marked in red. All the slots of $S$, with the oldest being $r_s = 1$, received a contribution because the number of free slots at the start of the cycle is an integer multiple of $n_t = 4$. Hence, the buffer becomes full at $r_c = 20$. The last 6 positions of the buffer $S$ are associated with incomplete sweeps, which dominate sections $D_{1,5}$ and $D_{2,5}$. The first $h = 20 - 7 + 1 = 14$ elements of $S$ are passed to the peak finding routine. Then, that portion of the buffer is cleared to make space for the next dedispersion cycle. The count of used slots decreases to $r_c = 20 - h = 6$. The start index $r_s$ of the buffer moves from position 1 to 15, the first incomplete sweep, that now becomes the oldest entry in the buffer.} \label{fig:timing1}
	\end{subfigure}
	\hfill
	\begin{subfigure}[t]{0.45\textwidth}
		\centering
		\includegraphics[width=\linewidth]{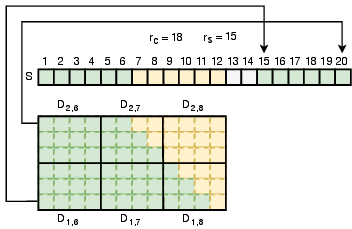} 
		\caption{\textbf{End of cycle 2}. The second cycle of dedispersion begins with 14 slots of the buffer available to host new sweeps, indexed form 1 to 14. Three more sections are generated and processed before the buffer cannot handle any further time bin intervals. Positions 13 and 14 remains unused: one would have to choose $B$ such that $(N - L + 1) \mod n_t = 0$ for the buffer to be completely utilised at every cycle. Slots 15 to 20, the oldest 6 of $S$ starting from $r_s$, receive all their missing contributions from side sweeps starting at sections $D_{1,6}$ and $D_{2,6}$. These are the same sweeps marked in yellow in Figure \ref{fig:timing1}. Slots 1 to 12 collect contributions to sweeps started during this cycle. When S becomes full again, slots 15 to 6 are searched for peaks, then cleared. Position 7 becomes the oldest used slot.} \label{fig:timing2}
	\end{subfigure}
	
	\vspace{1cm}
	\begin{subfigure}[t]{\textwidth}
		\centering
		\includegraphics[width=0.45\linewidth]{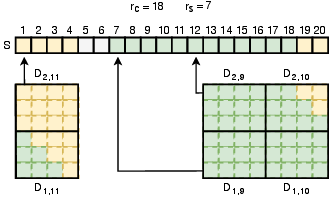} 
		\caption{\textbf{End of cycle 3}. The buffer state falls in the same pattern at the end of all cycles after the first one. Elements $r_s$ to $((r_s + h - 1) \mod N)$ are complete sweeps, the first $L - 1$ being the continuation of sweeps started during the previous cycle. Elements $((r_s + h)\mod N)$ to $((r_s + h + L - 2) \mod N)$ are incomplete sweeps, and the last two slots are unused.} \label{fig:timing3}
	\end{subfigure}
	\caption{\textbf{Example execution of the dedispersion algorithm using a ring buffer.} Depicted in this figure is the state of algorithm at the end of the first three dedispersion cycles. A cycle ends when the buffer $S$ becomes full and ready for downstream processing. During a dedispersion cycle, the \texttt{transient\_search} algorithm iterates over sections of a dynamic spectrum as they are produced by an upstream process. In this example there are 4 time bins in a time interval covered by a section. The full bandwidth is divided into 6 frequency channels, 3 per section. To ease the illustration of the example, the dispersive delay for each channel is chosen to be $2$ time bins, resulting into a sweep length $L= 7$. The number $B$ of additional slots in $S$ is arbitrarily set to 13. The resulting size of $S$  is $N = L + B = 20$. Slot indexes are displayed on top of the array. Every time the ring buffer becomes full, the first $h = r_c - L + 1$ elements starting from position $r_s$, coloured in green, form a short time series that is ready to be searched for transient signals.  Sections contributing to the dedispersed time series are shown below the ring buffer, aligned to $S$ according to the time bins they represent. A green cell in the ring buffer stores the accumulated intensity value of a complete sweep. The path of such sweeps is also coloured in green throughout the crossed sections. The yellow colour is associated with incomplete sweeps that occupy the latest entries of the buffer. Gray cells are not associated to any sweep due to $n_t$ not being always a multiple of $N - r_c$ at the beginning of a cycle.}
	\label{fig:ring_buffer}
\end{figure*}

% Ok, what do you want to say here?
% The figure shows three different stages or batch processing
% Well, define what a dedispersion cycle is. That is, when the buffer gets full.
% So, the figure shows the state of the algorithm and buffer at the end of three consecutive dedispersion cycles, starting from the very beginning when the buffer is empty. 

\section{Generalisation to multiple pixels and DMs}\label{sec:generalisation}

Only a single dynamic spectrum, partitioned in sections, is required to define the algorithm. It is the degenerate scenario of one-pixel images. Realistically, the number of pixels varies from tens of thousands to millions, and there are just as many dynamic spectra. The quantity of images that can be held in memory limits the number of frequency and time bins simultaneously available during processing. Then, the discretized observing frequency band $\mathcal{F}$ and time $\mathcal{T}$ is divided in manageable intervals $\mathcal{F}_i$  and $\mathcal{T}_j$ of size $n_f$ and $n_t$, respectively. A whole observation can be now imaged and dedispersed incrementally using STRIDE.

An image set $I_{i,j}$ contains all images corresponding to frequency channels in $\mathcal{F}_i$ and time bins in $\mathcal{T}_j$.  It encodes sections $D^{x, y}_{i, j}$, 
\begin{equation*}
(x, y) \in \{1, \ldots, X\} \times \{1, \ldots, Y\},
\end{equation*}

of the same time and frequency intervals within $XY$ different dynamic spectra. Given an image set $I_{i,j}$, the set of sweeps to be computed is the same for every pixel $(x, y)$ because the corresponding sections share the same frequency channels and time bins.  Algorithm \ref{algo:dedisp} is modified to enclose the two calls to \texttt{compute\_partial\_sweep} within for loops over all pixels. An example of such modification is shown in Algorithm \ref{algo:dedisp_multipix}. The output array $S$ is extended with enough memory to hold time series for every $(x, y)$ pair.

\begin{algorithm}
	\For{$t$ \textbf{in} $\mathcal{T}_j$}{
		$t_s \gets T(in_f, t, 0)$\;
		\textbf{if} $t_s < 1$ \textbf{then} \textbf{continue}\;
		\For{$(x, y)$ \textbf{in} $\{1, \ldots, X\} \times \{1, \ldots, Y\}$}{
			$S[x, y, t_s] \pluseq $ compute\_partial\_sweep($D^{x, y}_{i,j}$, $t$, $in_f$, $0$)\;
		}
	}
	\caption{An additional for loop is added to account for multiple pixels when computing top sweeps. The same is done for side sweeps, not illustrated here for brevity.}
	\label{algo:dedisp_multipix}
\end{algorithm}

% Generalisation to multiple DM in a section

The additional notation and algorithmic logic required to handle an arbitrary number $M$ of DM trials is a natural and minimal extension of the mathematical formalism and algorithm listings introduced in earlier sections. The delay $\tau$ and delay-dependent functions $\delta$, $g$, and $T$ become also a function of DM. A sweep is now identified with a 4-element tuple $(DM, f, t, d)$, and so on. Because DM trials are independent of one another, there are just as many resulting dedispersed time series the memory allocation for $S$ must account for. Algorithm \ref{algo:dedisp} is extended with an additional outmost for loop over DM trials whose values determine sweep paths across the section. Ring buffer slots where contributions to those sweeps are placed in are also determined by the associated DM values.

Multiple DM values imply varying sweep lengths. In this case, $L$ is chosen to be the maximum sweep length resulting from the largest DM. If the first $h$ longest sweeps in the buffer completed in the current iteration, then the shorter ones must have too. In fact, more may be ready for analysis. However, time series associated to different DMs are still processed at the same length $h$ to allow for a practical implementation. Enough memory must be allocated to the output array for it to also hold a de-dispersed time series for each DM trial. The final dimensions of $S$ are $N \times X \times Y \times M$.

\section{Parallelisation strategy} \label{sec:parallelisation}

Brute force dedispersion is computationally expensive, characterised by a time complexity of $O(TFM)$. The cost grows by orders of magnitude in imaging-based dedispersion where the algorithm is run for hundred thousands of pixels. Furthermore, data movement between computer memory and the processor becomes a significant contribution to the total runtime. An implementation of the algorithm must optimise memory access patterns and adopt compute parallelisation strategies to make the approach feasible in practice.

The algorithm can be parallelised at multiple abstraction levels. The solution presented in this section is the one adopted to extend the BLINK pipeline \citep{DiPietrantonio2025, GPUImager} with dedispersion and peak finding capabilities.  At the highest level, all image sets belonging to the same time interval can be dedispersed at the same time. It is achieved by distributing among different CPU threads the iterations of the inner for loop over frequency channel intervals of Algorithm \ref{algo:ring_buffer_algo}. Each CPU thread is assigned a GPU to dedisperse the associated image sets. The GPU kernel implementing Algorithm \ref{algo:dedisp} distributes pixels over GPU threads, whereas the DM trials, and top and side sweeps are iterated over serially by each thread. This choice reduces the complexity of the GPU implementation, guarantees a sufficient degree of parallelism, and optimises memory access and cache utilisation.

The input data layout is such that intensity information of the same time bin and frequency channel across all sections are closely packed in memory. A GPU thread warp, which is a set of GPU threads executing the same instructions in lock step, exploits data locality by accessing neighbouring pixels while integrating over the same sweep path for all directions in the sky. In doing so, memory accesses are combined in high-bandwidth memory transactions to L2 cache and GPU global memory.

The output buffer $S$ must be accessible by all GPUs for them to store their contributions to the total sweep intensities. Hence, it makes use of managed memory whereby the allocation resides on the main memory of the system while remaining accessible to threads of all GPUs. Atomic operations are required because multiple GPUs may write to the same memory address.

\section{Experimental results} \label{sec:results}

The STRIDE algorithm is designed to overcome the challenges of transient searches within low-frequency archival observations of the MWA Voltage Capture System (VCS; \cite{Tremblay2015}). A GPU-accelerated implementation of the algorithm extends the BLINK imaging pipeline with transient detection capabilities. The implementation of STRIDE, which closely resembles the listings in this work, is publicly available on GitHub\footnote{\url{https://github.com/PaCER-BLINK-Project/dedispersion}}.

\begin{table}
	\centering
	\begin{tabular}{|l|l|}
		\toprule
		\textbf{Property} & \textbf{Value} \\
		\bottomrule
		Number of tiles & 128 \\\hline
		Central frequency & 154.25\,MHz \\\hline
		Bandwidth & 30.72\,MHz \\\hline
		Number of channels & 768 \\\hline 
		Channel resolution & 40\,kHz \\\hline 
		Integration time & 20\,ms \\\hline
		Image size & 1024$^2$\,px \\\hline
		Pixel size & 0.007$^2$\,deg$^2$ \\\hline
	    DM & 50 to 60\,pc\,cm$^{-3}$ \\\hline
		Maximum dispersive delay & 4.25\,s \\\hline
		Number of images & 163,200 \\\hline
		Image volume & 684.5 GB \\\hline
	\end{tabular}
	
	\caption{\textbf{BLINK imaging and transient search parameter specification for the experiment.} The high time and frequency resolution, together with the dispersive delay, result in more than hundred thousand images to be considered at any given time to survey the imaged FoV for transient events. The memory demand to hold all of them cannot be met by the current generation of GPU hardware.}
	\label{tab:crab_requirements}
\end{table}

A 20-minute MWA VCS Phase 2 Extended observation whose FoV encompasses the position of B0531+21, also known as the Crab pulsar, was chosen as a test case. The goal is to correctly dedisperse and detect Crab pulses. Voltage data are stored in 24 files for each second of observation. Each file corresponds to a 1.28\,MHz coarse channel, then fine-channelised to 32, 40\,kHz channels by the BLINK correlator. The MWA Extended configuration comprises long baselines of up to 5\,km, ensuring positional accuracy of sources through high angular resolution and good quality images to facilitate software validation. Table \ref{tab:crab_requirements} summarises the imaging parameters for a BLINK transient search in the chosen observation. The time resolution selected for the experiment is 20\,ms, resulting in 50 time bins per observing second. Hence, a natural choice is to map input files to image sets of size $n_f = 32$, $n_t = 50$.

The number of images required to cover the entire 30.72\,MHz band and the time span of the 4.25\,s dispersive delay results in a data volume of 684.5\,GB. Running the BLINK pipeline on the entire observation generates 46 million images. However, the software is not constrained by these numbers because it images and dedisperses one second and one coarse channel of observation at a time. The streaming nature of the STRIDE algorithm and the adoption of a ring buffer reduce the minimum memory requirement by 97.9\%, down to 14.4\,GB, occupied by the ring buffer to support the dedispersion process. The figure accounts for the maximum dispersive delay and 2\,s worth of additional memory slots for each of the time series in the ring buffer. That is, $L = \lceil4.25\,s / 0.02\,s \rceil = 213$ and  $B = 2\,s / 0.02\,s  = 100$.

% Discuss the execution environment 
The observation was split in two 10-minute slices and independently searched for radio pulses using the BLINK transient search pipeline. The DM trial values are between 50 and 60\,pc\,cm$^{-3}$, in steps of 1\,pc\,cm$^{-3}$. The BLINK pipeline was run on 2 GPU nodes equipped with 4 AMD MI250X GPUs each, part of the Setonix supercomputer \citep{setonix} hosted at the Pawsey Supercomputing Research Centre\footnote{https://pawsey.org.au} in Perth. The total execution time across the two nodes was 24.7 hours. Dividing the figure by $1024^2$, the number of pixels, gives a normalised runtime of 0.085\,s per pixel. The compute hardware recorded an energy consumption of 27.4\,kWh.  

\begin{figure}
\includegraphics[width=\linewidth]{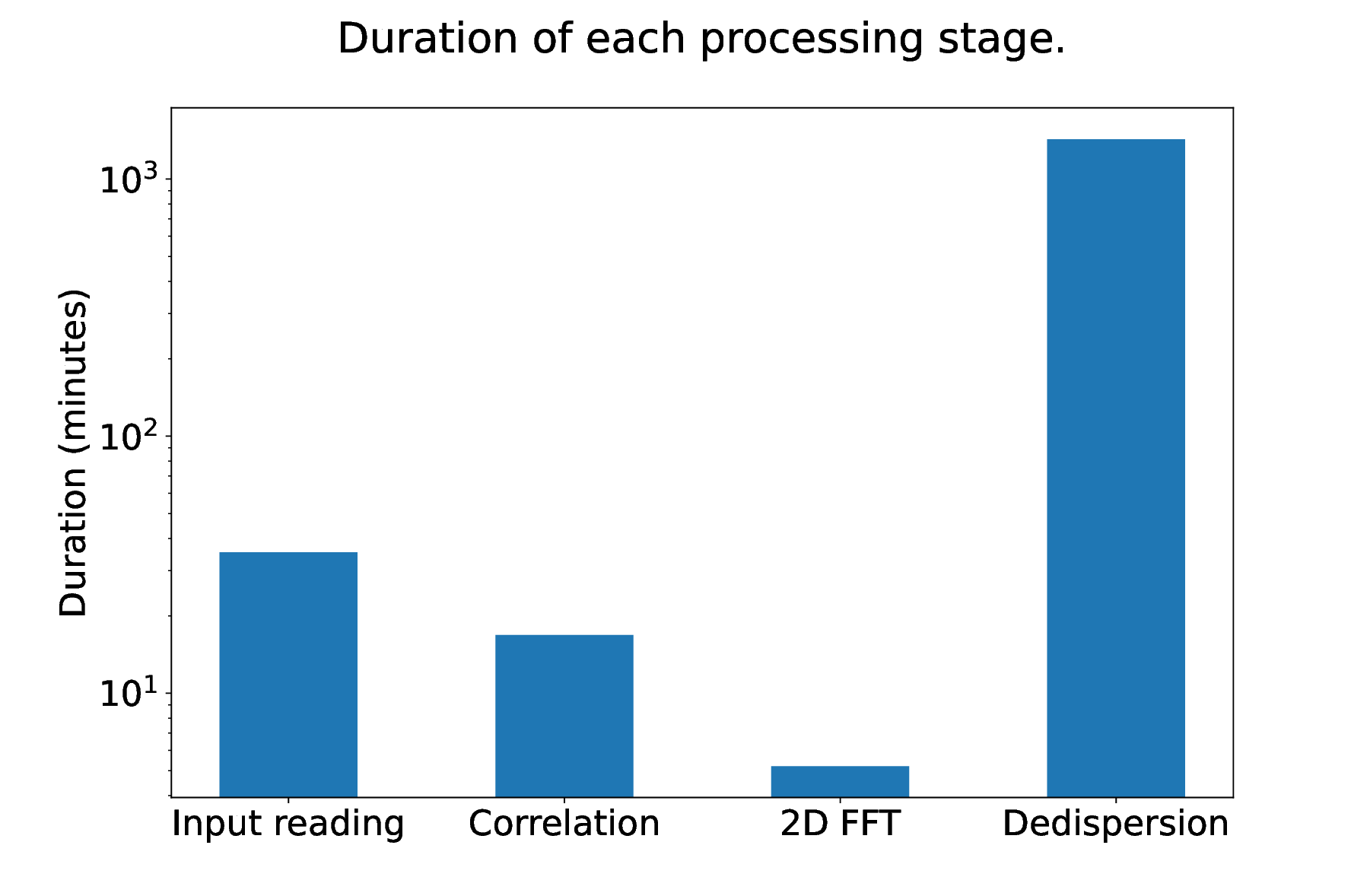}
\caption{\textbf{Execution times of the BLINK processing stages.} The four most expensive computational steps are reading the voltage files, cross correlation of voltage time series, inverse 2D Fourier transform of visibilities into sky images, and dedispersion through the STRIDE algorithm. }
\label{fig:processing_time}
\end{figure}

Figure \ref{fig:processing_time} reports the execution times of the main computational steps of the BLINK pipeline, in logarithmic scale. The dedispersion stage accounts for 96.1\% of the total runtime, or 23.8 hours, and it is three orders of magnitude more expensive than the cumulative execution time of the 2D FFT step. The latter sits at 0.35\%, equal to 312 seconds, as the least time consuming operation. The input ingestion accounts for 2.37\%, 35.5 minutes, of the BLINK runtime due to data residing on the SSD-based fast storage of the \texttt{/scratch} filesystem. Finally, the correlation kernel takes 16.8 minutes, making up the remaining 1.38\%.

The BLINK pipeline produced 138,786 raw candidates with a signal-to-noise (SNR) ratio above the threshold of 7. These were then clustered by $(x, y)$ coordinates, DM, and time bin to identify a single pulse event. The same pulse can in fact be detected in multiple adjacent pixels, at several DMs around the true one, and at consecutive time bins. A second clustering operation groups together events at close DM and location to identify the associated source. The hierarchical clustering process revealed two sources at DM values of 57 and 51\,pc\,cm$^{-3}$. These were detected at multiple positions within the FoV. For each source, the cluster containing the brightest detections gives a first-order estimation of its true position, whereas the others are regarded as side lobe detections.

The source at DM 57 \,pc\,cm$^{-3}$ records 358 events at location RA = 05:34:31.13, DEC = +21:53:26.5. The recorded maximum SNR is 125. The Crab pulsar is listed to be 0.12 degrees away with a DM of 56.77\,pc\,cm$^{-3}$ on the ATNF Pulsar Catalogue\footnote{\url{https://www.atnf.csiro.au/research/pulsar/psrcat/}}. The second source detected with a DM 51, located at RA = 05:28:51.22, DEC = +21:51:06.0, counts 78 events with a maximum SNR of 39. The closest known source is B0525+21 at a distance of 0.15 degrees away and characterised by a DM of 50.86\,pc\,cm$^{-3}$.

Figure \ref{fig:main_findings} visually summarises the main findings of the experiment. The dirty image of the FoV reported in Figure \ref{fig:crab_image} shows the Crab nebula as the brightest continuum radio source present in the observation. The image also contains several side-lobe detections of the nebula distributed radially around the real source. Figure \ref{fig:detection_map} is a detection map where each candidate is shown as a dot that encodes its position within the FoV, and its signal-to-noise and DM value through the dot size and colour, respectively. To keep the figure readable, only candidates with SNR above 10 are shown. They are 49,346, or 35.6\% of the total. Most of these overlap with the true source and side lobe artifacts of the Crab nebula. The dot colours correspond to DM values clustered around 57\,pc\,cm$^{-3}$, another confirmation these are pulses associated with B0531+21. A distinct cluster on the top-left of the map highlights the second source at a well-defined DM of 51\,pc\,cm$^{-3}$. Figure \ref{fig:crab_time_series} and \ref{fig:b0525_time_series} show the dedispersed time series of the brightest pulse of each detected source. The dynamic spectrum of the Crab pulse with strongest SNR, computed in post-processing, is shown in Figure \ref{fig:crab_ds}.

% Discuss what we save as output.

\begin{figure*}
	\centering
	\begin{subfigure}[t]{0.45\textwidth}
		\centering
		\includegraphics[width=0.80\linewidth]{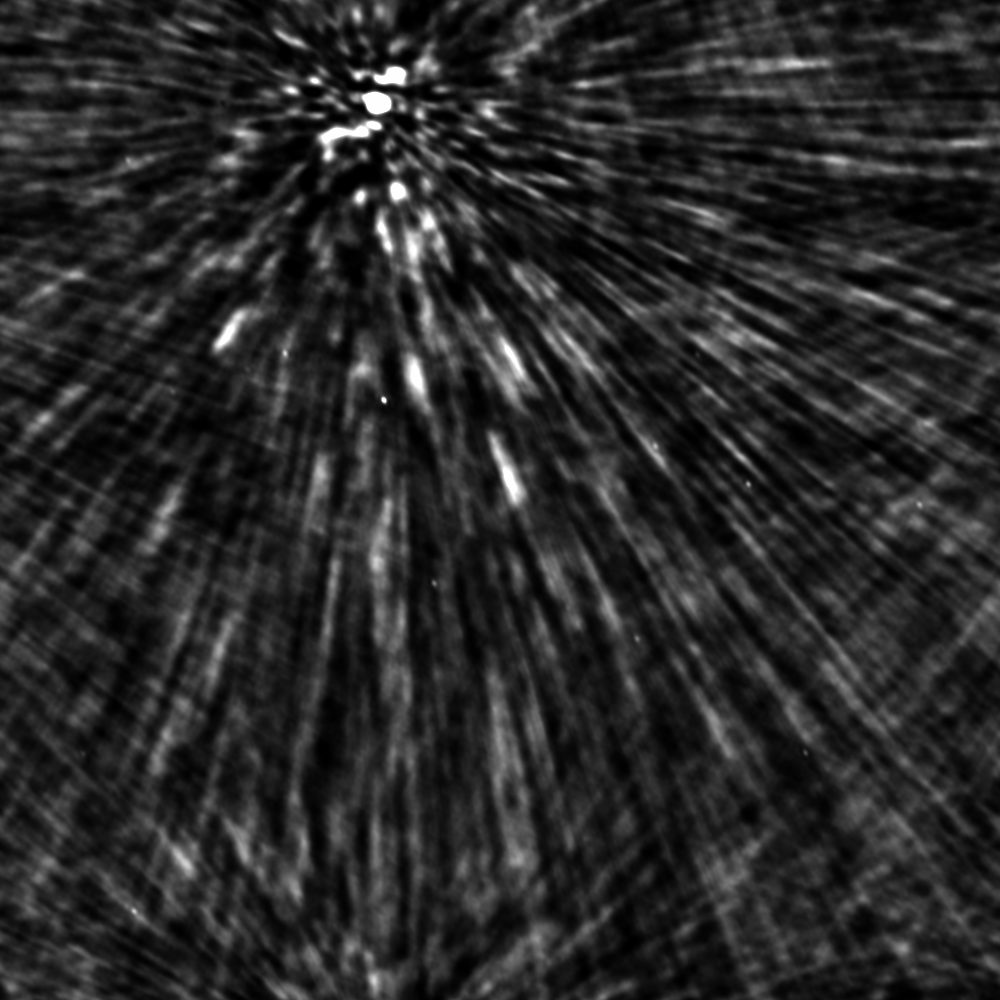} 
		\caption{A dirty image of the FoV reveals the Crab nebula as a bright persistent radio source towards the upper edge of the image. The other bright spots leading up to it are the effect of the side lobes of the primary beam that are sensitive to the nebula due to its strong brightness. Point-like sources can also be seen spread throughout the image.}
		\label{fig:crab_image}
	\end{subfigure}
	\hfill
	\begin{subfigure}[t]{0.50\textwidth}
		\centering
		\includegraphics[width=\linewidth]{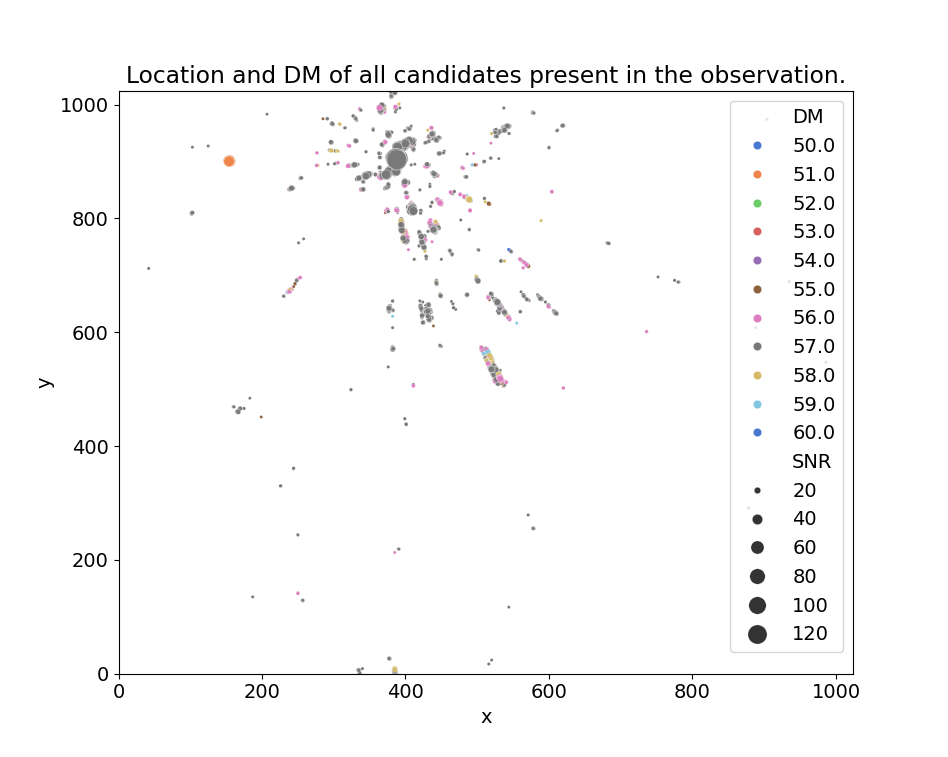} 
		\caption{Candidates with an SNR value above 10 are displayed in a scatter plot, referred to as a detection map. For each candidate, its $(x, y)$ coordinates correspond to its position in the image, the dot size is a function of the SNR, and the colour encodes the DM value. Two main clusters are present in the map. The gray-coloured cluster characterised by a DM of 57 and an SNR of around 120 matches the position of the Crab nebula, as shown in Figure \ref{fig:crab_image}. Smaller clusters distributed around it are side lobe detections of the same source, also at nearby DMs. A distinct orange cluster to the far left of the figure highlights another set of detections at DM 51 with an SNR of around 40. The positions of the clusters correspond to pulsars B0531+21 and B0525+21, respectively.}
		\label{fig:detection_map}
	\end{subfigure}

\begin{subfigure}[t]{0.48\textwidth}
	\centering
	\includegraphics[width=\linewidth]{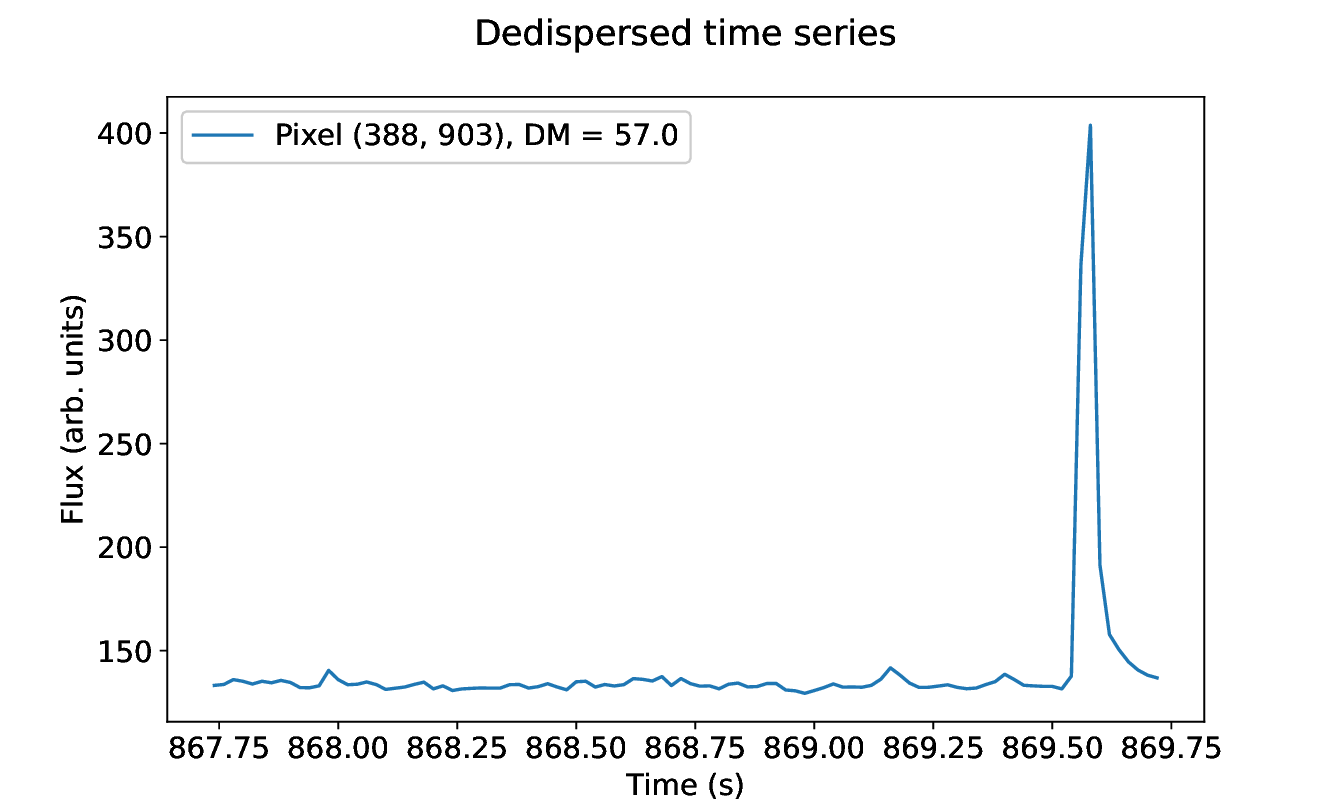} 
	\caption{The brightest detection of B0531+21 with SNR equal to 125. The same pulse is shown in Figure \ref{fig:crab_ds}, dispersed in time while crossing the observed frequency band. The median value of the time series reflects the brightness of the host nebula.}
	\label{fig:crab_time_series}
\end{subfigure}
	\hfill
\begin{subfigure}[t]{0.48\textwidth}
	\centering
	\includegraphics[width=\linewidth]{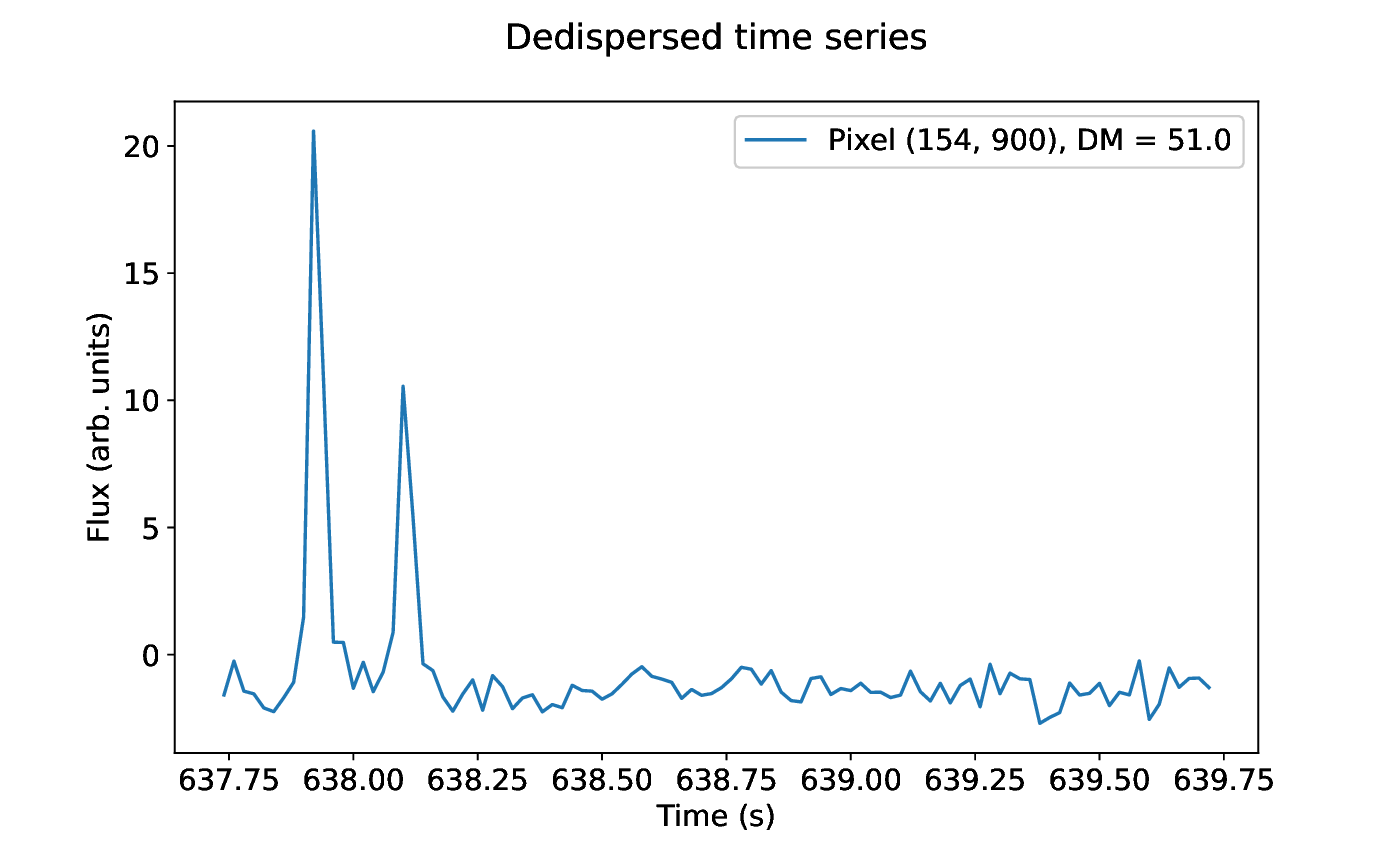} 
	\caption{The brightest detection of B0525+21 with an SNR in excess of 39, followed by another strong pulse.}
	\label{fig:b0525_time_series}
\end{subfigure}
	\caption{\textbf{The detection of B0531+21 and B0525+21 with the BLINK pipeline.}  The single pulse transient search revealed more than a hundred thousand candidate pulses above a SNR threshold of 7. The majority are due to two bright pulsars at 1.2 degrees distance from one another. The first is the rotating neutron star B0531+21 within the Crab nebula, and the second is the pulsar B0525+21. The post-processing software allows the inspection of candidates using a variety of diagnostic plots. These include a detection map displaying candidates at their positions within the FoV, and dedispersed time series plots highlighting flux peaks that triggered the detections.}
	\label{fig:main_findings}
\end{figure*}
% Talk about the implementation

% Context: the BLINK pipeline

\section{Conclusion} \label{sec:conclusion}

The computational cost of imaging is asymptotically lower than the one of beamforming, presenting an opportunity to accelerate the discovery of new millisecond-scale radio transients with widefield interferometric searches. However, the computational advantage is only accessible after having dealt with the significant memory requirements, in the order of up to terabytes, implied by the large number of high time and frequency resolution images. The problem is especially present at frequencies below 300\,MHz, where the dispersive delay affecting a signal is up to tens of seconds. Table \ref{tab:crab_requirements} accurately characterises the requirements of a transient search with the MWA at a relatively low DM range. The amount of computer memory necessary steeply increases when the DM is in the order of hundreds of units, as it is the case for FRBs.

STRIDE tackles the problem by being the first algorithm to decompose the dedispersion problem over the time axis, hence eliminating the requirement of holding on to all images associated with a time span as long as the maximum dispersive delay. Furthermore, by adopting a streaming processing strategy, the dedispersed time series are searched for peaks as soon as enough data points are produced. Only the time series within the ring buffer triggering a candidate detection are saved for later inspection. 

Imaging-based searches for fast transients at low frequencies remain challenging. The dispersive delay still dictates the minimum number of time bins that must be reserved for each time series within the ring buffer. The wide FoV and high angular resolution of modern interferometers result in more than a million pixels, directly affecting the size of the ring buffer and the number of compute operations to be executed. For instance, the image size of $1024^2$ pixels adopted in the presented experiment only covers approximately $9\%$ of the 600\, deg$^2$ wide FoV at 150\, MHz that characterises the MWA in Extended configuration. The wide range of DM values has the same implication.

Despite still requiring some compromises on the search parameter space, STRIDE is the first dedispersion algorithm that makes imaging-based, millisecond-scale transient searches at low frequencies feasible. While initially designed for the MWA, it provides a general strategy employable in a transient search back-end of any interferometer.

\section*{Acknowledgements}

Authors would like to thank Clancy W. James for providing useful comments during the writing of this paper.

This work was supported by resources provided by the Pawsey Supercomputing Research Centre’s Setonix Supercomputer, with funding from the Australian Government and the Government of Western Australia.

This scientific work uses data obtained from Inyarrimanha Ilgari Bundara / the Murchison Radio-astronomy Observatory. We acknowledge the Wajarri Yamaji People as the Traditional Owners and native title holders of the Observatory site. Establishment of CSIRO's Murchison Radio-astronomy Observatory is an initiative of the Australian Government, with support from the Government of Western Australia and the Science and Industry Endowment Fund. Support for the operation of the MWA is provided by the Australian Government (NCRIS), under a contract to Curtin University administered by Astronomy Australia Limited. 

This work was supported by the Australian SKA Regional Centre (AusSRC), Australia’s portion of the international SKA Regional Centre Network (SRCNet), funded by the Australian Government through the Department of Industry, Science, and Resources (DISR; grant SKARC000001). AusSRC is an equal collaboration between CSIRO – Australia’s national science agency, Curtin University, the Pawsey Supercomputing Research Centre, and the University of Western Australia.
%\end{acknowledgement}

\bibliographystyle{elsarticle-harv} 
\bibliography{bibliography.bib}

% PASA uses footnotes, not endnotes. \endnote in this template will behave like \footnote; and \printendnotes will not output anything.
% \printendnotes
% \printbibliography

\appendix

\section{Proofs} \label{sec:proofs}

\begin{theorem}
	There are  $\delta (f) - 1$ side sweeps for each frequency channel $f$in $D_{i,j}$.
\end{theorem}

\begin{proof}
	A side sweep enters channel $f$ at some time $t_a$,  $t_a  < (j - 1)n_t + 1$, spends $\delta(f)$ time bins crossing the channel, and leaves the same at some exit time $t_b$, $t_b = t_a + \delta(f) - 1$, $t_b \geq (j - 1)n_t + 1$. Assume the earliest value for $t_a$ is less than $t_a^{lo}$, 
	
	\begin{equation}
		t_a^{lo} = (j - 1)n_t + 1 - (\delta(f) - 1).
	\end{equation}
	Then, its corresponding exit time $t_b$ is
	\begin{align}
		t_b &= t_a + \delta(f) - 1 \\
		&  < t_a^{lo} + \delta(f) - 1 \\
		&< (j - 1)n_t + 1 - \delta(f) + 1 + \delta(f) - 1\\
		& <  (j - 1)n_t + 1,
	\end{align}
	
	contradicting the assumption that this is a side sweep of channel $f$. Conversely, assume $t_a > t_a^{hi}$, $t_a^{hi} = (j - 1)n_t$: then this is not a side sweep by definition. It follows that the entry time $t_a$ of a side sweep for channel $f$ must satisfy the constraint $t_a^{lo} \le t_a \le t_a^{hi}$. There are
	
	\begin{align}
		t_a^{hi} + 1 - t_a^{low}	&= (j - 1)n_t + 1 - ((j - 1)n_t + 1 - (\delta(f) - 1))\\
		& = \delta(f) - 1
	\end{align}
	
	valid time bins for $t_a$, and hence just as many side sweeps entering $D_{i,j}$ through channel $f$.
	
\end{proof}

\section{Author ORCID Identifiers}

\begin{itemize}
	\item Cristian Di Pietrantonio: \url{https://orcid.org/0000-0002-7175-9079}
	\item Marcin Sokolowski: \url{https://orcid.org/0000-0001-5772-338X}
	\item Christopher Harris: \url{https://orcid.org/0000-0001-5237-2250}
	\item Danny C. Price: \url{https://orcid.org/0000-0003-2783-1608}
	\item Randall Wayth: \url{https://orcid.org/0000-0002-6995-4131}
\end{itemize}

\newpage
\onecolumn
\section{List of symbols}

Table \ref{tab:list_of_symbols} provides a summary of the symbols used throughout the paper with a reference to the section or equation where they are first defined.

\begin{table*}[h!]
	\begin{tabularx}{\textwidth}{l X r}
		\hline
		\textbf{Symbol} & \textbf{Definition} & \textbf{Reference}\\\hline\hline
		$A_{i,j}$ & The set of all top sweeps of section $D_{i,j}$. & Equation \ref{eq:top_sweep_set} \\
		$B$ & The number of additional memory slots for the ring buffer to hold a dedispersed time series of significant length. & Equation \ref{eq:buffer_length} \\
		$B_{i,j}$ & The set of all side sweeps of section $D_{i,j}$. & Equation \ref{eq:side_sweep_set} \\
		$D(f, t)$ & A dynamic spectrum illustrating signal intensity as a function of time and frequency. & Section \ref{sec:dedisp_problem} \\
		$D_{i,j}(f, t)$ & The 2D section of a dynamic spectrum covering the $i$-th frequency channel interval and the $j$-th time bin interval. & Section \ref{sec:math_framework} \\
		$D_{i,j}^{x,y}(f, t)$ & The 2D section of the dynamic spectrum of pixel $(x, y)$ covering the $i$-th frequency channel interval and the $j$-th time bin interval. & Section \ref {sec:generalisation} \\
		$d$ & The offset into the intra-channel dispersive delay. & Equation \ref{eq:time_series_full_dedisp} \\
		$d_{start}$ & The number of time bins the sweep has already spent in $f_{top}$. & Algorithm \ref{algo:partial_sweep} \\
		$d_{max}$ & The maximum number of intensity values \texttt{compute\_partial\_sweep} accumulates in one channel.  & Algorithm \ref{algo:partial_sweep} \\
		
		$DM$ & The dispersion measure. & Section \ref{sec:introduction} \\
		$\Delta f$ & The frequency resolution, or width, of a frequency channel. & Section \ref{sec:dedisp_problem}\\
		$\Delta t$ & The time resolution of a time bin. & Section \ref{sec:dedisp_problem}\\
		$\delta(f)$ & The intra-channel dispersive delay in frequency channel $f$, in units of time bins. & Equation \ref{eq:channel_delay}\\
		$F$ & The number total of frequency channels in $D$. Also used as index of the top frequency channel. & Section \ref{sec:motivation} \\
		$\mathcal{F}$ & The set of all frequency channels. & Section \ref{sec:dedisp_problem}\\
		$\mathcal{F}_i$ & The $i$-th frequency channel interval. &  Equation \ref{eq:freq_interval}\\
		$f_{top}$ & The top frequency channel of a section. & Section \ref{sec:algorithm} \\
		$g(f)$ & The time delay, in units of time bins, between channel $F$, included, and channel  $f$, excluded. & Equation \ref{eq:cumulative_delay}\\
		$h$ & The number of values in the dedispersed time series corresponding to completed sweeps and that are ready to be searched for peaks. & Algorithm \ref{algo:ring_buffer_algo}\\
		$I_{i,j}$ &  The set of all images corresponding to the frequency channels in $\mathcal{F}_i$ and the time bins in $\mathcal{T}_j$. & Section \ref{sec:generalisation}\\
		$K$ & A product of physical constants contributing to the dispersive delay.  & Equation \ref{eq:time_delay}  \\
		$L$ & Length of a sweep in units of time bins. & Equation \ref{eq:sweep_length} \\
		$M$ & The number of DM trials. & Section \ref{sec:generalisation} \\
		$N$ & The number of slots in the ring buffer dedicated to a time series. & Section \ref{sec:ring_buffer} \\
		$n_f$ & The number of frequency channels encompassed by a section of a dynamic spectrum. & Section \ref{sec:math_framework} \\
		$n_t$ & The number of time bins encompassed by a section of a dynamic spectrum. & Section \ref{sec:math_framework} \\
		$\nu_f$ & The bottom frequency of the $f$-th frequency channel. & Section \ref{sec:dedisp_problem}\\
		$(p, q)$ & The coordinate pair used to loop over a dispersion path within a dynamic spectrum. & Equation \ref{eq:total_intensity_as_triple_sum}\\
		$r_c$ & The number of consumed slots in the ring buffer of a given time series. & Algorithm \ref{algo:ring_buffer_algo} \\
		$r_s$ & The index of the earliest slot of the ring buffer of a given time series. &  Algorithm \ref{algo:ring_buffer_algo} \\
		$S$ & The time series output array as produced by the algorithm. & Algorithm \ref{algo:ring_buffer_algo}\\
		$S(t)$ & The dedispersed time series of a signal. & Equation \ref{eq:time_series_full_dedisp}\\
		$s$ & The accumulation variable for the intensity of a partial sweep. & Algorithm \ref{algo:partial_sweep} \\
		$s(f, t, d)$ & The sweep passing through channel $f$ at time bin $t$, having already spent the previous $d$ time bins in the same channel $f$. & Section \ref{sec:math_framework}\\
		$T$ & The total number of time bins in $D$. & Section \ref{sec:motivation}\\
		$\mathcal{T}$ & The set of all time bins. & Section \ref{sec:dedisp_problem}\\
		$\mathcal{T}_j$ & The $j$-th time bin interval. &  Equation \ref{eq:time_interval}\\
		$t_s$ & The start time bin of the sweep being currently computed. & Algorithm \ref{algo:dedisp}\\
		$T(f, t, d)$ & The index of the time bin when sweep $s(f, t, d)$ enters the top frequency channel $F$. & Equation \ref{eq:start_time}\\
		$\tau(\nu_{lo}, \nu_{hi})$ & The time delay between the arrival times of a signal at frequencies $\nu_{lo}$ and $\nu_{hi}$. & Equation \ref{eq:time_delay}\\
		$X$, $Y$ & The dimensions of an image. & Section \ref{sec:motivation}\\
		$(x, y)$ & The coordinate pair of an image pixel. & Section \ref{sec:generalisation}\\
		\hline
	\end{tabularx}
	\caption{List of mathematical symbols used throughout the work, each accompanied with a short description and a reference to its definition in the paper.}
	\label{tab:list_of_symbols}
\end{table*}

\end{document}